\documentclass[pra,twocolumn,letterpaper]{revtex4}
\usepackage{graphicx}
\usepackage{amsthm}
\usepackage{amsmath}
\usepackage{amsfonts}
\usepackage{verbatim}

\newcommand{\beq}{\begin{equation}}
\newcommand{\eeq}{\end{equation}}
\newcommand{\ket}[1]{|#1\rangle}
\newcommand{\bra}[1]{\langle#1|}
\newcommand{\proj}[1]{\ket{#1}\bra{#1}}
\newcommand{\braket}[2]{\langle #1 | #2 \rangle}

\newcommand{\Tr}{{\rm Tr}}

\newcommand{\psucc}{p_{\textup{succ}}}

\newtheorem{theorem}{Theorem}
\newtheorem{corollary}{Corollary}
\newtheorem{observation}{Observation}
\newtheorem{lemma}{Lemma}

\begin{document}

\title{Linear-optics realization of channels for single-photon multimode qudits}

\author{Marco Piani$^1$, David Pitkanen$^1$,Rainer Kaltenbaek$^{1,2}$,Norbert L{\"u}tkenhaus$^1$}

\affiliation{
$^1$Institute for Quantum Computing \& Department of Physics
and Astronomy, University of Waterloo, 200 University 
Avenue West, N2L 3G1 Waterloo, Ontario, Canada 
\\
$^2$Faculty of Physics, University of Vienna, Boltzmanngasse 5, A-1090 Vienna, Austria}

\date{\today}

\begin{abstract}
We propose and theoretically study a method for the stochastic realization of arbitrary quantum channels on multimode single-photon qudits. In order for our method to be undemanding in its implementation, we restrict our analysis to linear-optical techniques, vacuum ancillary states and non-adaptive schemes, but we allow for random switching between different optical networks. With our method it is possible to deterministically implement random-unitary channels and to stochastically implement general, non-unital channels. We provide an expression for the optimal probability of success of our scheme and calculate this quantity for specific examples like the qubit amplitude-damping channel. The success probability is shown to be related to the entanglement properties of the Choi-Jamio{\l}kowski state isomorphic to the channel.
\end{abstract}

\maketitle

\section{Introduction}

The most general transformation a quantum state can undergo is described by a quantum channel. For example, it may correspond to a controlled manipulation of a quantum system for some final aim---like in quantum information-processing protocols~\cite{nielsen00a}---or it may represent an unwanted interaction with the environment. While in the first case implementing the respective quantum channel is of direct practical interest, in the second case one may still be interested in the implementation of the channel for the sake of understanding the role of noise, and how to counteract it, in real-world implementations of quantum information-processing protocols. It is worth remarking that striking effects in quantum information processing (QIP), e.g., the super-activation of the quantum capacity of channels~\cite{GraemeSmith09262008}, involve non-trivial noisy channels. 

Quantum optics is one of the best established physical architectures for QIP~\cite{knill01a,Mattle1996a,bouwmeester97a,Pan2003a,prevedel07a}. It has the advantage that the carriers of information---photons---interact naturally weakly with the environment, so that \emph{real} noise is low. This makes \emph{simulating} noise possible in a very controlled way. The workhorse of optical experiments is the manipulation via linear optical elements, such as beam splitters and phase-shifters. The linear-optics realization of channels has appeared in a number of works for specific cases. For example, random-unitary channels are common in experiments on decoherence-free and unitarily recoverable subspaces~\cite{Prevedel2007b,Schreiter:2009p1583} and in the realization of mixed states \cite{Amselem2009a,Lavoie2010a}. The simplest non-trivial example of a channel that is not random-unitary is perhaps given by the qubit amplitude-damping channel \cite{nielsen00a}. The counting statistics of this channel have been simulated using linear optics \cite{Almeida:2007p1504}, and a stochastic linear-optical implementation with a fixed success probability of $50 \% $, independent of the value of the damping parameter, has been suggested \cite{Qing:2007p1454}.

In this article, we propose a linear-optics scheme for the stochastic exact realization of an arbitrary channel for single-photon multimode qudits. Under constraints motivated by the ease of experimental realization, our scheme achieves an optimal probability of success. An interesting result is that such a success probability is related to the entanglement properties of the Choi-Jamio{\l}kowski state isomorphic to the channel~\cite{jamiolkowski72a,choi75a}. This connection allows us to apply results in entanglement theory~\cite{Horodecki:2009p2733} to the quite different problem of channel realization.

Our results provide an optimal strategy for the realization of arbitrary channels, an important building block in experimental studies of QIP. In the specific case of the qubit amplitude-damping channel, our scheme provides a significantly higher efficiency than alternative schemes~\cite{Qing:2007p1454} without leaving the subspace of the encoding of the input state. In contrast to \cite{Almeida:2007p1504}, this allows us to further process the output of the channel.

The paper is structured as follows. In Section \ref{sec:definitions}, we provide definitions, fixing both the framework and the notation. In Section \ref{sec:theproblem}, we illustrate in detail the problem we consider, that is, the realization of a quantum channel with a fixed set of tools. In Section \ref{sec:solution}, we provide a scheme to realize any channel perfectly albeit only stochastically. In Section \ref{sec:entanglement}, we relate the optimal success probability of the method proposed to the entanglement properties of the Choi-Jamio{\l}kowski state isomorphic to the channel of interest. In Section \ref{sub:Bounds}, we use this relation to provide bounds on the probability of success, both in the specific case of qubits, for which we are able to give analytic bounds, and qudits. In Section \ref{sec:examples}, we apply our technique to two examples, one being the qubit amplitude damping channel. Finally, we conclude and discuss possible future venues to investigate.

\section{Definitions and framework}
\label{sec:definitions}

The state of a quantum system may change over time due to some internal dynamics, to an interaction with its
environment or to a measurement performed on it by an observer.
Any physical transformation a quantum system can experience can be modeled as a quantum channel $\Lambda:\rho_{\text{in}}\mapsto\rho_{\text{out}}$.
Every channel acting on a system $S$ admits a dilation, that means, it can be realized as some unitary interaction with an ancilla $E$, which is subsequently discarded:
\[
\Lambda[\rho_S]=\Tr_E(U_{SE}\,\rho_{S}\otimes\sigma_{E}\,U_{SE}^\dagger),
\]
with $\sigma_E$ the initial state of the ancilla~\cite{nielsen00a}.
More abstractly a quantum channel can be defined as a completely positive trace-preserving linear map.  Each channel can be represented in the form $\Lambda[\rho]=\sum_iA_i\rho A_i^\dagger$, where $\{A_i\}$ is a set of \emph{Kraus operators} fulfilling the trace-preserving condition, $\sum_i A_i^\dagger A_i=I$. The Kraus representation of a channel is not unique. For instance, if $\{A_i\}$ forms a \emph{Kraus decomposition} of a channel $\Lambda$, the relation $B_i = \sum_j u_{ij} A_j$, assuming $u_{ij}$ are the elements of a unitary matrix, will define a new decomposition $\{B_i\}$ for $\Lambda$~\cite{nielsen00a}. 

We will frequently find the notion of operator norm useful in our discussions of quantum channels. Since we will always work in finite dimensions, the operator norm $\|A\|_\infty$ of $A$ corresponds to the largest singular value of $A$. An operator is an admissible Kraus operator---that is, it can be considered as part of some valid Kraus-operator set---as long as $\|A\|_\infty\leq1$. Any set of linear operators that satisfy the completion relation $\sum_{i=1}^{k}A_{i}^{\dagger}A_{i}=I$ will constitute a valid quantum channel.

In this paper, we will be interested in optical quantum systems. Each
mode of an optical system is associated to a basis of Fock states
$| n\rangle$, where $n=0,1,2...$ denotes the number of photons
in the mode. The creation and annihilation operators, $a^\dagger$ and $a$,
respectively, provide a convenient notational framework for describing Fock states because of the relations $a| n\rangle=\sqrt{n}| n-1\rangle$, $a^{\dagger}| n\rangle=\sqrt{n+1}| n+1\rangle$, so that 
$| n\rangle=\left(a^{\dagger}\right)^{n}/\sqrt{n!}|0\rangle$. These operators have commutation relations $[a_{i},a_{j}^{\dagger}]=\delta_{ij}$, $[a_{i}^{\dagger},a_{j}^{\dagger}]=0$
and $[a_{i},a_{j}]=0$, where the indices $i$ and $j$ denote the optical mode and $\delta_{ij}$ is the Kronecker delta.

The most common optical elements that are used in experiments for the manipulation of optical modes are beam splitters and phase shifters. Optical networks that are composed only of instances of these two elements are referred to as passive linear devices. Linear (quantum) optics is the part of quantum optics that, apart from the initial generation of entangled photon pairs and single-photon detection, deals only with passive linear devices \cite{knill01a}.
Any unitary transformation $U$ acting on $d$ optical modes and preserving the total photon number can conveniently be described by the way it transforms the creation operators of the
modes:
\beq
\label{eq:unitary_lin_optics}
a_i^{\text{out}\dagger}=\sum_ju_{ij}a_j^{\text{in}\dagger},
\eeq
where $u_{ij}$ are the elements of a unitary matrix. A transformation can be realized by linear optics if and only if it is of this kind.

A phase shifter is an optical element that acts on a single mode as $U a^{\dagger}U^{\dagger}=e^{i\phi}a^{\dagger}$. A beam splitter acts on two optical modes at a time and can be
described by
\begin{equation}
\label{eq:BS}
(u_{ij})=\left(\begin{array}{cc}
\cos\theta & -e^{i\phi}\sin\theta\\
e^{-i\phi}\sin\theta & \cos\theta\end{array}\right).
\end{equation}
Any unitary that acts on $d$ modes preserves the total photon number if and only if
it can be implemented using these two devices~\cite{reck94a}. 

\section{The problem}
\label{sec:theproblem}

A qudit can be encoded by using one photon in $d$ optical modes. An arbitrary logical state can be written as $\ket{\psi}=\sum_{i=1}^d\psi_i\ket{i_{\text L}}$, with a logical basis $\left\{ | i_{L}\rangle\right\} _{i=1}^{d}$, where $\ket{i_{\text L}}=a_i^\dagger\ket{0}$. We call this kind of encoding a \emph{$d$-rail encoding}.
This encoding is convenient when the interactions
are limited to linear optics, because any unitary operation \eqref{eq:unitary_lin_optics} can be
performed on the creation operators using linear optics, and under
this encoding the basis states of a single qudit and the creation operators transform
identically.

We are interested in the simulation of an arbitrary quantum channel $\Lambda$ that acts on a qudit, using only passive linear optics. What we want is a realization of $\Lambda$ on the $d$-rail qudit, such that the logical subspace---the encoding---is mapped onto itself. This allows for further processing of the output of the channel.

We will refer to the channel to be realized as the \emph{logical channel}, to distinguish it from physical channels that evolve the state of the modes without necessarily preserving the logical subspace.

As we noted earlier, we can always represent a channel in the form of a dilation where the channel is realized via the unitary interaction of the system with ancillary modes. We will limit ourselves to linear-optics evolution. For the sake of the ease of experimental implementation, we will assume several other reasonable restrictions: (i) to limit the number of photons that need to be generated, we only introduce ancillary modes that are initially in the vacuum state; (ii) in order to prevent the necessity of using expensive feed-forward mechanisms (Pockels cells and high-speed high-voltage switches --- see, e.g., \cite{prevedel07a, Biggerstaff2009a}), we do not allow adaptive schemes; (iii) we will restrict ourselves to photon-number measurements, although it will actually turn out that commonly used threshold detectors suffice.

When we consider the dilation representation of a channel, we can imagine that the final trace over the ancillary space corresponds to a measurement of the ancilla, whose result is discarded. If we assume that the ancilla starts in the (vacuum) state $\ket{0}_E$, then we have
\begin{multline*}
\Lambda[\rho_S]\\
\begin{aligned}
&=\Tr_E(U_{SE}\rho_{S}\otimes\sigma_{E}U_{SE}^\dagger)\\
			&=\sum_k \Tr_E(U_{SE}\rho_{S}\otimes\sigma_{E}U_{SE}^\dagger M^k_E)\\
			&=\sum_{jk} \left(\bra{j}_E \sqrt{M_k^E} U_{SE}\ket{0} \right) \rho_S  \left( \bra{j}_E \sqrt{M_k^E} U_{SE}\ket{0}\right)^\dagger
\end{aligned}
\end{multline*}
with $\{M_k\}$, $M_k\geq0$, $\sum_k M_k=I$ a POVM on the ancilla system $E$, and $\{\ket{j}\}$ an eigenbasis for $\sigma_E$.
With our constraints---vacuum input ancillas and linear optics evolution---measuring the vacuum on the output ancillas is the only result that leaves the system within the encoding.
This can be seen easily considering the action of the linear optics unitary $U_{SE}=U_{\text{LO}}$ on initial states $\ket{i_{\text{L}}}\ket{0}_{\text{E}}$, $i=1,\dots,d$. We will consider $d+e$ modes, with the first $d$ used for the encoding, and the remaining $e$ constituting the ancilla system $E$. Then we have:
\[
\begin{aligned}
U_{\text{LO}}\ket{i_{\text{L}}}_S\ket{0}_{\text{E}}&=U_{\text{LO}}a_i^\dagger\ket{0}_S\ket{0}_{E}\\
			&=\sum_{j=1}^{d+e} u_{ij} a_j^\dagger\ket{0}_S\ket{0}_{E}\\
			&=\left(\sum_{j=1}^{d} u_{ij} \ket{j_{\text{L}}}\right)\ket{0}_{E}\\
			&+\ket{0}_S\left(\sum_{j=d+1}^{d+e} u_{ij} a_j^\dagger \ket{0}_{E}\right).
\end{aligned}
\]
From this expression it is evident that if we perform a photon-number measurement on the output ancillary modes and we obtain a result different from the vacuum, then the encoding is lost. 
The reason for this is that linear optics preserves the photon number
and the initial state of the system $\ket{i_{\text{L}}}\ket{0}_{\text{E}}$
only has one photon in it. If the photon is measured in the ancilla, then the initial
state $\ket{i_{\text{L}}}$ will be mapped out of the encoding to the
vacuum, independently of which output ancilla mode the photon is measured in.

Therefore, under the constraints that we have imposed, the only logical channels that can be realized deterministically must have a single Kraus operator. Such channels are necessarily unitary transformations, as it can be seen by the trace-preservation condition $A^{\dagger}A=I$.

\section{The solution: stochastic implementation}
\label{sec:solution}

In this section, we will first see that \emph{any} single logical Kraus operator (i.e., any Kraus operator of the logical channel) can be realized stochastically. Later we will introduce a further resource, randomness, and the ability to switch---according to such randomness---among different optical networks, and we will show that then any logical channel can be realized, albeit only stochastically.

\subsection{Implementation of a logical Kraus operator}

For any logical Kraus operator $A$ that we want to apply to the input state, it is possible to construct an optical network such that $A$ will correspond to the transformation of the logical state if the output ancillary modes are detected to be in the vacuum state, given that they were in the vacuum state before the channel. Every Kraus operator has a singular value decomposition $A=VSU$, where $U$ and $V$ are unitaries
and the matrix $S$ is positive and diagonal, with diagonal elements $0\leq s_{i}\le1$ that correspond to the singular values of $A$.
As unitary rotations can be realized deterministically on the encoding, in order to prove that A can be realized under our constraints, it is sufficient to prove that any diagonal matrix S can be realized
(see Figure
\ref{fig:BSrealization}).

This is proven to be possible by considering the action \eqref{eq:BS} of a beamsplitter on two modes.
If the first mode, with creation operator $a^\dagger$, belongs to the encoding and the second mode is an
ancilla---which means it starts in the vacuum---then the transformation
that results when the vacuum is measured on the ancilla state effectively
realizes the mapping $a^{\dagger}\mapsto\cos\theta a^{\dagger}$.
 Since the angle $\theta$ is arbitrary, we can simply implement any diagonal logical 
Kraus operator $S$ by using $d$ ancillary modes and $d$ beamsplitter, choosing the angles $\theta_i$ such that $s_{i}=\cos\theta_{i}$.%
\begin{figure}
\includegraphics[scale=0.5]{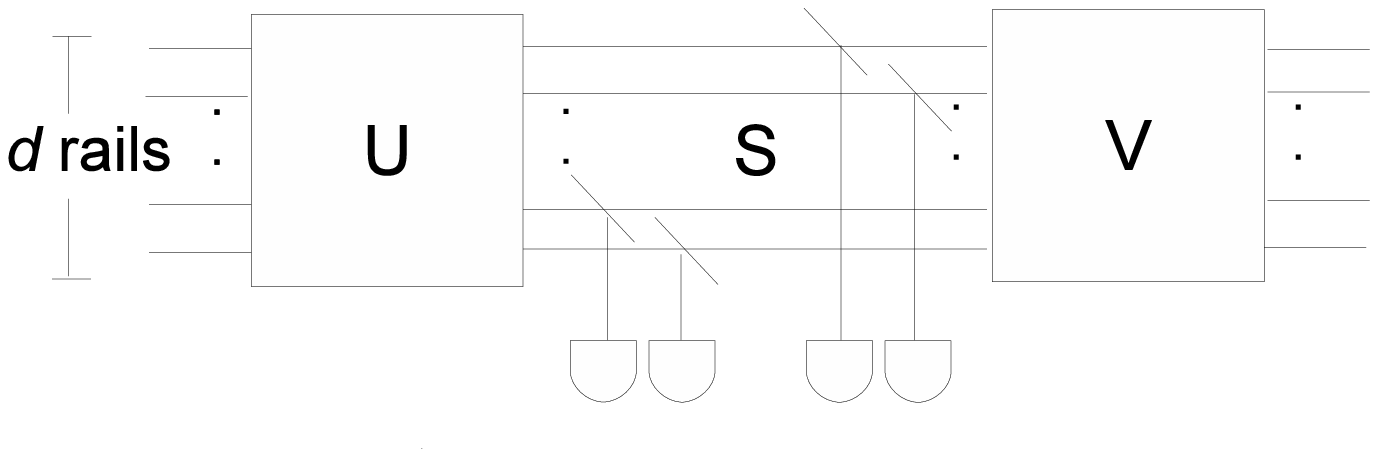}
\caption{The diagram describes an optical circuit for a channel that would
realize the Kraus operator $A=VSU$ as well as a Kraus operator that
would map the encoding to the vacuum. The boxes represent optical
arrays that perform the unitary that labels them. The $S$ transformation then consists
of a set of beamsplitters, one for each mode, whose transmission coefficients
are matched to the singular values of the matrix $S$.}
\label{fig:BSrealization}
\end{figure}

\subsection{Perfect but stochastic implementation of an arbitrary logical channel\label{sub:Impmnt-any-ch}}

A logical channel $\Lambda$ that we may want to apply on the encoding will in general have a Kraus decomposition $\{A_{i}\}_{i=1}^{n}$, with $n\geq1$. Therefore, by using a fixed linear optical network in the framework defined in Section \ref{sec:theproblem} it will not be possible in general to simulate the channel, as only one logical Kraus operator can be realized per fixed optical network.

We will circumvent this problem by realizing individually the various Kraus operators  $A_{i}$, $i=1,\ldots,n$, in this way being able to preserve the encoding for each $A_i$. Roughly speaking, by randomly applying the different Kraus operators the logical channel $\Lambda$ will be realized.
Of course, this is possible only by allowing the linear optical network to change. We will introduce the possibility of switching among various optical networks---one for each $A_i$---according to a probability distribution $\{p_i\}$. Each fixed optical network that we will introduce to realize the Kraus operator $A_i$ will itself correspond to a quantum channel $\Gamma_i$ (see Figure
\ref{fig:BBox}).
This ``average realization'' of the logical channel will anyway be stochastic, because in the implementation of any $A_i$ that is not unitary there will necessarily be a finite probability of ending up outside the encoding, which corresponds to finding the input photon in the output ancillary modes.

One important point is that, given the additional degree of freedom due to the choice of the probability distribution $\{p_i\}$, it is possible to consider the realization of a rescaled version $\tilde{A}_i$ of $A_i$  rather than exactly $A_i$. Of course each $\tilde{A}_i$ must be a valid Kraus operator, i.e., $\|\tilde{A}_i\|_\infty\leq1$. We will use this rescaling degree of freedom to maximize the success probability for the realization of the channel.

If we postselect on finding the output ancillary modes in the vacuum state, and if we choose the probability distribution $\{p_i\}$ and the $\tilde{A}_i$ operators such that $\sqrt{p}_{i}\tilde{A}_{i}=\sqrt{p_{\text{succ}}}A_{i}$ for all $i$ and for some $0\leq p_{\text{succ}}\leq1$, then the logical input state $\rho$ will be mapped into the (unnormalized) logical state
\[
\sum_i p_i \tilde{A}_i \rho \tilde{A}_i^\dagger=p_{\text{succ}}\sum_i A_i\rho A_i^\dagger.
\]
This will happen with probability $\Tr(\sum_i p_i \tilde{A}_i \rho \tilde{A}_i^\dagger)=p_{\text{succ}}$, and thus the logical channel $\Lambda$ will be stochastically implemented with probability $p_{\text{succ}}$ (independent of the input $\rho$).

Given that we want the channel to be realized perfectly, the figure of merit we care about is the probability of success $p_{\text{succ}}$, which we want to be maximal. One possible choice for the distribution $\{p_i\}$ and the operators $\tilde{A}_i$ is trivially $p_i=1/n$ and $\tilde{A}_i=A_i$; this choice leads to a probability of success $p_{\text{succ}}=1/n$. This strategy is independent of the properties of the Kraus operator $\{A_i\}$ for the particular channel $\Lambda$, and depends only on the number of Kraus operators. As such, one can expect it to be non-optimal, and it certainly is in the case of a random-unitary channel
\[
\Lambda[\rho]=\sum_iq_iU_i\rho U_i^\dagger,
\]
with $\{U_i\}$ unitaries and $\{q_i\}$ a probability distribution. Indeed, in this case an obvious better choice---and actually optimal---is $p_i=q_i$, $\tilde{A}_i=U_i$, for all $i$, such that $p_{\text{succ}}=1$.

The following theorem provides the optimal choice of the probability distribution $\{p_i\}$ and of the operators $\tilde{A}_i$'s to maximize $p_{\text{succ}}$, for any fixed Kraus decomposition $\{A_i\}$.

\begin{figure}
\includegraphics[scale=0.5]{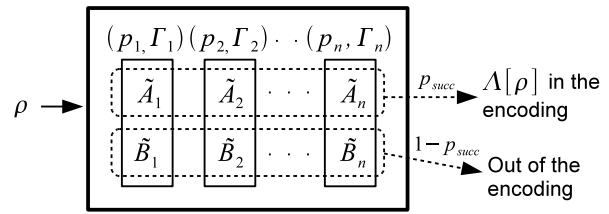}
\caption{Pictorial representation of our scheme. Each solid rectangle represents a channel. The most external box
is a mixture of the $n$ channels $\Gamma_i$ inside of it. Each of these inner channels
corresponds to a linear-optics setup and for our scope its action on the encoding can be completely described without loss of generality by two Kraus operators, $\tilde{A}_i$ and $\tilde{B}_i$. Each $\tilde{A}_i$
preserves the $d$-rail encoding, while the Kraus operators $\tilde{B}_i$ map an encoded state out of the encoding. If the condition $\sqrt{p_{i}}\tilde{A}_{i}=\sqrt{\psucc}A_{i}$, for all $i$, is met, the overall result of randomly switching among the channels $\Gamma_i$ according to the probability distribution $\{p_i\}$ is that of realizing the target logical channel $\Lambda$ with probability $\psucc$ (independent of the input $\rho$). }
\label{fig:BBox}
\end{figure}

\begin{theorem} Given the Kraus decomposition $\{A _{i}\}$
for the channel $\Lambda$, the optimal probability of success for its realization is
\beq
p_{\textup{succ}}(\{A_{i}\})=\frac{1}{\sum_i \|A_{i}\|_{\infty}^{2}}.
\eeq
This can be achieved by the choice $p_{i}=\frac{\|A_{i}\|_{\infty}^{2}}{\sum_{j}\|A_{j}\|_{\infty}^{2}}$
and $\tilde{A}_{i}=\frac{1}{\|A_{i}\|_{\infty}}A_{i}$, for all ${i}$.
\end{theorem}

\begin{proof}
From the condition $\sqrt{p_{i}}\tilde{A}_{i}=\sqrt{\psucc } A_{i}$, for all $i$,
one finds
$p_i\geq p_i\|\tilde{A}_i\|^2=\psucc\|A_i\|^2_\infty$, where we used the fact that $\|\tilde{A}_i\|_\infty \leq 1$, because each $\tilde{A}_i$ must be a proper Kraus operator. Summing over $i$ and using $\sum_{i}p_{i}=1$, one arrives at $p_{\text{succ}}\leq1/\sum _i \|A_{i}\|_{\infty}^{2}$.
The probability distribution and Kraus operators in the statement of the theorem saturate the inequality.
\end{proof}

Thus, the maximal probability of simulating the channel adopting the Kraus decomposition $\{A_i\}$ in our scheme is the inverse of $\sum_i \|A_{i}\|_{\infty}^{2}$.
This quantity will in general depend on the specific Kraus decomposition. By optimizing over all Kraus decompositions we have the following.

\begin{corollary}
\emph{(Optimal probability of success)}
In our scheme, the optimal probability of success in the implementation of $\Lambda$ is
\beq
\label{eq:psucc}
p_{\textup{succ}}(\Lambda)=\max_{\{A_i\}}\frac{1}{\sum_i\|A_{i}\|_{\infty}^{2}},
\eeq
where the maximization is over all Kraus decompositions $\{A_i\}$ of the channel $\Lambda$.
\end{corollary}
For convenience in the analysis to follow, we define the  \emph{stochasticity} of a channel as
\beq
\label{eq:stochasticity}
\sigma(\Lambda)=\min_{\{A_i\}} \sum_i\|A_{i}\|_{\infty}^{2},
\eeq
where the minimization is over all Kraus decompositions $\{A_i\}$ of the channel $\Lambda$, so that
\[
p_{\textup{succ}}(\Lambda)=\frac{1}{\sigma(\Lambda)}.
\]
The name ``stochasticity'' is justified by the fact that the larger $\sigma(\Lambda)$, the lower the probability of a successful realization of the channel.

We remark that any specific Kraus decomposition will give an upper (lower) bound on the stochasticity (optimal probability of success).

\section{Relation with entanglement measures\label{sec:Connect-Ent}}
\label{sec:entanglement}

The optimal success probability $p_{\textup{succ}}(\Lambda)$ of implementing a channel
is clearly just a property of the channel itself. Therefore it appears natural to look for a representation
of the channel that is independent of any specific Kraus decomposition. This can be done by considering the Choi-Jamio{\l}kowski isomorphism.

The latter is a one-to-one mapping between maps and operators~\cite{jamiolkowski72a,choi75a}.
The isomorphism---explicitly in the direction from maps to operators---is defined as
\beq
\begin{split}
\label{eq:CJ}
J(\Lambda) & = \left(\Lambda\otimes I\right)\left[\psi_{d}^{+}\right]\\
 & = \frac{1}{d}\sum_{i,j=1}^{d}\Lambda(| i\rangle\langle j|)\otimes| i\rangle\langle j|,
 \end{split}
 \eeq
for some fixed choice of a maximally entangled state
\beq
\label{eq:maxent}
|\psi_{d}^{+}\rangle=\frac{1}{\sqrt{d}}\sum_{i=1}^{d}| i\rangle| i\rangle.
\eeq

For our purpose, the interesting observation is that pure ensemble decompositions $\{p_i,\psi_i\}$ of the Choi-Jamio{\l}kowski state $J_\Lambda$ isomorphic to a channel $\Lambda$ are in one-to-one correspondence with Kraus decompositions of $\Lambda$. This follows from the the fact that for any vector $\ket{\bar{\psi}}\in\mathbb{C}^d\otimes\mathbb{C}^d$ there is an operator $A^{\bar{\psi}}$ such that
\beq
\label{eq:CJ_linop}
|\bar{\psi}\rangle=(A^{\bar{\psi}}\otimes I)|\psi_{d}^{+}\rangle.
\eeq
Here, the bar in $\ket{\bar{\psi}}$  denotes that the vector need not be normalized. In general, unless it is specified to the contrary with use of the bar notation, all states $\ket{\psi}$ are assumed to be normalized.
Thus,
\begin{equation}
\begin{split}
J(\Lambda) &=  \sum_{i}p_{i}|\psi_{i}\rangle\langle\psi_{i}|\\
 & = \sum_{i}|\bar{\psi}_{i}\rangle\langle\bar{\psi}_{i}|\\
 & = \sum_{i}(A^{\bar{\psi}_{i}}\otimes I)|\psi_{d}^{+}\rangle\langle\psi_{d}^{+}|({A^{\bar{\psi}_{i}\dagger}}\otimes I),
 \end{split}
 \end{equation}
for $|\bar{\psi}_{i}\rangle=\sqrt{p_i}|\psi_{i}\rangle=(A^{\bar{\psi}_{i}}\otimes I)|\psi_{d}^{+}\rangle$.

As we have seen, without the use of randomness the only channels that can be realized deterministically are unitaries, and using randomness and switching among optical networks we can extend this result only to random unitaries. Thus, we have that
the only channels that can be realized deterministically in our framework are those whose Choi-Jamio{\l}kowski state admits an ensemble consisting only of maximally entangled states.

One then expects that channels whose probability of realization is high admit Kraus decompositions that are close to random-unitary. In turn this would mean that their Choi-Jamio{\l}kowski states admit ensemble decompositions that are highly entangled. We will show that this intuition is correct.
 
The relation \eqref{eq:CJ_linop} implies
\beq
\label{eq:relnorm}
\|A^{\bar{\psi}}\|^2_\infty=d \times \lambda_{\text{max}}(\bar\psi),
\eeq
if  we consider the Schmidt decomposition $\ket{\bar{\psi}}=\sum_i\sqrt{\lambda_i}\ket{i}\ket{i}$, with $\lambda_i\geq 0$, $\sum_i\lambda_i=\braket{\bar{\psi}}{\bar{\psi}}$, and $\lambda_{\text{max}}=\max_i\{\lambda_i\}$. Thus, we find for the stochasticity
\begin{subequations}
\label{eq:connection1}
\begin{align}
\sigma(\Lambda)&=\min_{\{A_i\}}\sum_i\|A_i\|_\infty^2\\
			&=d\min_{\{p_i,\psi_i\}}\sum_ip_i\lambda_{\text{max}}(\psi_i)\label{eq:stocschmidt}\\
				&=d\left(1-\max_{\{p_i,\psi_i\}}\sum_ip_i(1-\lambda_{\text{max}}(\psi_i))\right)\\
				&=d\left(1-\max_{\{p_i,\psi_i\}}\sum_ip_iE_G(\psi_i)\right)
\end{align}
\end{subequations}
where we used \eqref{eq:relnorm} to move from the minimization over Kraus decompositions for $\Lambda$ to the minimization over ensemble decompositions for $J(\Lambda)$. The quantity
\[
E_G(\psi)=1-\lambda_{\text{max}}(\psi)=1-\max_{\alpha,\beta}|\braket{\alpha,\beta}{\psi}|^2,
\]
where the maximum is taken with respect to factorized pure states $\ket{\alpha,\beta}=\ket{\alpha}\ket{\beta}$, is the \emph{geometric measure of entanglement} for a bipartite pure state~\cite{wei03a}. More generally, for a multipartite pure state, the geometric measure of entanglement is defined as $E_G(\psi)=1-\max_{\phi_{\text{sep}}}|\braket{\phi_{\text{sep}}}{\psi}|^2$, with $\phi_{\text{sep}}$ a fully separable state. In the bipartite case, it coincides with the entanglement measure $E_2$, which was defined in \cite{Vidal:1999p955} as one of a whole family of entanglement measures. The geometric measure of entanglement has received a good deal of attention \cite{Wei:2004p2746,Plenio:2007p1778} because of its intuitive---even in the multipartite case---geometric interpretation as maximal overlap of the state of interest with a fully separable state,  and because of its connections to other well-known entanglement measures, like relative entropy of entanglement \cite{PhysRevLett.78.2275,PhysRevA.57.1619}. In the bipartite qudit case we are interested in here, one sees immediately that
\beq
\label{eq:geobounds}
0\leq E_G(\psi) \leq 1-\frac{1}{d}.
\eeq
The lower bound is achieved for a factorized pure state, while the upper bound corresponds to a maximally entangled state like the one in Eq.~\eqref{eq:maxent}.

The geometric measure of entanglement is extended to the mixed-state case by the usual convex-roof construction~\cite{Uhlmann:1998p1153}:
\beq
E^{\cup}_G(\rho)=\min_{\{p_i,\psi_i\}}\sum_ip_iE_G(\psi_i),
\eeq
where we use $^\cup$ to stress that the resulting quantity is convex on the set of mixed states.

The standard convex-roof is defined in terms of the ensemble containing, on average, the \emph{minimum} amount of entanglement as quantified, in this case, by the geometric measure of entanglement for pure states. Eq. \eqref{eq:connection1}  involves instead the ensemble containing on average the \emph{maximum} amount of entanglement. This corresponds to the \emph{concave}-roof construction
\beq
\label{eq:EGcap}
E_G^\cap(\rho)=\max_{\{p_i,\psi_i\}}\sum_ip_iE_G(\psi_i),
\eeq
where we use $^\cap$ to stress that in this way we are defining a concave function on the set of mixed states.

For the sake of comparison with quantities better known in literature, let us mention that in the same way in which the \emph{entanglement of formation}~\cite{bennett96b} $E_F(\rho^{AB})=\min_{\{p_i,\psi^{AB}_i\}}\sum_ip_iS(\rho_i^{A})$, with $\rho^A=\Tr_B(\proj{\psi^{AB}})$ and $S(\sigma)=-\Tr(\sigma\log_2\sigma)$ the von Neumann entropy of a state $\sigma$, is the paradigmatic example for a convex roof construction, the \emph{entanglement of assistance}~\cite{DiVincenzo:1999p1604}
\beq
\label{eq:Eass}
E_{\text{a}}= \max_{\{p_i,\psi^{AB}_i\}}\sum_ip_iS(\rho_i^{A})
\eeq
is the paradigmatic example for a concave-roof construction.

From \eqref{eq:connection1} it follows that the stochasticity is given by
\beq
\label{eq:relent}
\sigma(\Lambda)=d(1-E^\cap_G(J(\Lambda)),
\eeq
and, as a result, the relation between the probability of success $\psucc(\Lambda)$ for our scheme to realize a channel $\Lambda$ and the entanglement properties of the related Choi-Jamio{\l}kowski state $J(\Lambda)$ can be expressed as
\beq
\label{eq:connection2}
\psucc(\Lambda)=\frac{1}{d\left(1-E^\cap_G(J_{\Lambda})\right)}.
\eeq
We remark that, because $E^\cap_G$ is a concave function on states, the probability of success $\psucc$ is a convex function on channels, i.e.,
\[
\psucc((1-q)\Lambda_1+q\Lambda_2)\leq (1-q) \psucc(\Lambda_1) + q \psucc(\Lambda_2),
\]
for $0\leq q\leq 1$. Of course, this could be concluded directly from \eqref{eq:psucc}.

\section{Bounds\label{sub:Bounds}}

The evaluation of the stochasticity \eqref{eq:stochasticity} for a given channel is in general
a non-trivial computational problem. The connection with entanglement
that was developed in Section \ref{sec:Connect-Ent}, more precisely Eq. \eqref{eq:relent}, shows that calculating
the stochasticity is equivalent to evaluating $E^\cap_G(J({\Lambda}))$.  In principle, this requires to check for all possible ensemble decompositions of $J({\Lambda})$, although one can use convexity arguments to restrict the search to ensembles of $r^2$ pure states for a Choi-Jamio{\l}kowski state of rank $r$, similarly to the case of entanglement of formation \cite{Uhlmann:1998p1153}. In this section we will be able to  provide analytic upper and lower bounds that do not require any search.

Entanglement of assistance and other concave-roof constructions have not been studied as well as convex-roof constructions. This is due to the fact that they are not entanglement measures~\cite{Plenio:2007p1778,Gour:2006p2799}. Nonetheless, they are of interest because, e.g., they capture some properties of multipartite entanglement. For example, the entanglement of assistance quantifies the average amount of entanglement that two parties---Alice and Bob---can share thanks to a measurement of a third party who holds the purification of the state. Thus, we will be able to make use of some results already derived in literature, in particular in \cite{DiVincenzo:1999p1604} and \cite{Laustsen:2003p1210}, to provide upper and lower bounds for the stochasticity $\sigma$ and the probability of success $\psucc$. 

We first start by illustrating the range over which $\psucc$ can vary, illustrating the best and worst cases. We then identify a simple bound based uniquely on the mathematical properties of the operator norm. As we will see, such a bound will turn out to be pretty useful in investigating the examples of Section \ref{sec:examples}. We then proceed to consider bounds based on the entanglement properties of the Choi-Jamio{\l}kowski state isomorphic to the channel of interest.

\subsection{Best and worst cases\label{sub:extreme-pnts}}

Given that $E_G$---and therefore $E_G^\cap$---satisfies \eqref{eq:geobounds}, it follows from \eqref{eq:connection2} that
\beq
\label{eq:extremes}
\frac{1}{d}\leq \psucc(\Lambda) \leq 1.
\eeq
As has been pointed out earlier, the upper bound in \eqref{eq:extremes} can only be achieved by random-unitary channels, whose Choi-Jamio{\l}kowski states can be written as convex combinations of maximally entangled states. The lower bound corresponds to $E^\cap_G(J(\Lambda))=0$, i.e., to the case where no ensemble for $J(\Lambda)$ contains any entangled state. Such an occurrence was considered in the context of the study of the entanglement of assistance in~\cite{DiVincenzo:1999p1604}, where it was proved that any state $\rho^{AB}$ with vanishing entanglement of assistance
must be of the form $\rho^{AB}=\proj{\alpha}\otimes\rho^B$ or $\rho^{AB}=\rho^A\otimes\proj{\beta}$. Given that we are not considering general bipartite states, but states that are isomorphic to channels via the isomorphism \eqref{eq:CJ}, for the first inequality in \eqref{eq:extremes} to be saturated it must be  $J(\Lambda)=\proj{\alpha}\otimes I/d$. The latter condition implies that the output of the channel is a pure state independent of the input, i.e., $\Lambda[\rho]=\Tr(\rho)\proj{\alpha}$.

It may seem strange that the channel that is almost the most trivial theoretically is the one
that is the most difficult to implement under our constraints. One can provide the following intuitive explanation. The output state must be independent of the input, but at the same time still be in the encoding. Thus, the output state must include the photon of the input encoding, because the ancillary modes are initially in the vacuum state. This can be accomplished in the following way. In the scheme proposed in Fig. \ref{fig:BSrealization}, a random rotation is first applied to the input. Subsequently, $d-1$ of the encoding modes are measured while the remaining one is transmitted---that is,  the transitivity of $d-1$ of the $d$ beam splitters is set to 0, while the remaining one is set to 1. Upon finding the vacuum in the measured modes, we know that the photon is in the only unmeasured mode, i.e., in some known logical basis state of the encoding. Then we can rotate such a state to the desired output state. Given the random rotation of the input, the probability that this procedure succeeds is exactly $1/d$, independent of the input.

\subsection{Triangle-inequality bound \label{sub:Bounds-2}}

By using the triangle inequality, it is straightforward to derive an upper limit on the success probability.
\begin{observation}
\label{obs:convbound} \emph{(Triangle-inequality bound)}
For any quantum channel $\Lambda$,
\begin{equation}
\label{eq:convbound}
\psucc(\Lambda)\leq \frac{1}{\|\Lambda(I)\|_{\infty}}.
\end{equation}
\end{observation}

\begin{proof}
If $\left\{ A_{i}\right\} $ is any Kraus decomposition for
the channel $\Lambda$ then we have for the stochasticity:
\begin{eqnarray*}
\sigma(\Lambda) & = & \min_{\left\{ A_{i}\right\} }\sum_i\|A_{i}\|_{\infty}^{2}\\
 & = & \min_{\left\{ A_{i}\right\} }\sum_{i}\|A_{i}A_{i}^{\dagger}\|_{\infty}\\
 & \ge & \min_{\left\{ A_{i}\right\} } \|\sum_{i}A_{i}A_{i}^{\dagger}\|_{\infty}\\
 & = & \|\Lambda(I)\|_{\infty},
 \end{eqnarray*}
 where the inequality is due to the triangle inequality, and the dependence on the choice of the Kraus decomposition is lost because $\sum_{i}A_{i}A_{i}^{\dagger}=\Lambda(I)$, for any Kraus decomposition of $\Lambda$.  
 \end{proof}

This bound proves that it is necessary for a channel to be unital in order for us to 
implement it deterministically using our scheme, because only for a unital channel $\|\Lambda(I)\|_{\infty}=1$. This is consistent with the already argued fact  that under our scheme only random-unitary channels can be deterministically 
implemented. The bound is easily evaluated, being independent of any particular Kraus decomposition. 

We remark that any choice of a specific Kraus decomposition provides a lower bound on the probability of success. If such a lower bound matches the upper bound in \eqref{eq:convbound}, then the given decomposition is proven to be optimal.

\subsection{Bounds based on entanglement properties of the Choi-Jamio{\l}kowski state\label{sub:Bounds-1}}

Now we will move to bounds that exploit the connection we observed between the success probability of our scheme and the entanglement properties of the Choi-Jamio{\l}kowski state $J(\Lambda)$.

\subsubsection{Qubit channels\label{sub:Analytic-Bounds}}

We will first focus on the qubit case. Not surprisingly, this is the case where we can employ most results from entanglement theory. In particular, we will mostly be concerned with the one of the most common entanglement measures, known as concurrence~\cite{Hill:1997p1217,wootters98a}. The concurrence of a pure two-qubit state can be expressed as
\[
C(\psi) = |\langle\tilde{\psi}|\psi\rangle|,
\]
where 
\[
|\tilde{\psi}\rangle = (\sigma_{y}\otimes\sigma_y)|\psi^{*}\rangle.
\]
with the complex conjugation taken in the computational basis and $\sigma_{y}=\begin{pmatrix}
0 & -i\\
i & 0\end{pmatrix}$. The definition extends to density matrices via the standard convex-roof construction:
 \[
C^\cup(\rho)=\min_{\{p_{i},\psi_{i}\}}\sum_{i}p_{i}C(\psi_{i}).
\]
It is straightforward to check that for a pure state the relation
\beq
\label{eq:relgeoconc}
E_G(\psi)=\frac{1}{2}\left(1-\sqrt{1-C(\psi)^{2}}\right)
\eeq
holds. In \cite{wei03a} it was argued that $C^\cup$ and $E_G^\cup$ are related by
$E_G^\cup=\frac{1}{2}\left(1-\sqrt{1-C^\cup(\rho)^2}\right)$.
We will instead be interested in the connection between $E_G^\cap$ and the concave-roof of the concurrence,
 \beq
 \label{eq:concconc}
C^\cap(\rho)=\max_{\{p_{i},\psi_{i}\}}\sum_{i}p_{i}c(\psi_{i}).
\eeq
The examples of Section \ref{sec:examples} will prove that the relation \eqref{eq:relgeoconc} does not hold for the concave-roof version of the two quantities. Nonetheless, in order to obtain easily computable bounds for $\psucc$, we will exploit the remarkable fact that there is a closed expression for both $C^\cup(\rho)$ and $C^\cap(\rho)$. For the former it reads~\cite{wootters98a}
\begin{eqnarray}
C^\cup(\rho) &=& \max\{0,\lambda_1-\lambda_2-\lambda_3-\lambda_4\}, \nonumber
\end{eqnarray}
where the $\lambda_i$'s are the eigenvalues of  $\sqrt{\sqrt{ \rho } \tilde{\rho} \sqrt{\rho}}$ in decreasing order, with the state $\tilde{\rho}=\sigma_{y}\otimes\sigma_{y}\rho^{*}\sigma_{y}\otimes\sigma_{y}$. For $C^\cap(\rho)$ instead it holds~\cite{Laustsen:2003p1210}:
\beq
C^\cap(\rho) = F(\rho,\tilde{\rho}), \label{eq:Ccapfidelity}
\eeq
with $F(\sigma,\tau)=\Tr\left(\sqrt{\sqrt{\sigma}\tau\sqrt{\sigma}}\right)$ the fidelity between two states $\sigma$ and $\tau$.
We start by providing the following lemma that relates $C^\cap(\rho)$ and $E_G^\cap(\rho)$.
\begin{lemma}
\label{lem:egeoconc}
Given any state $\rho$ of two qubits, the following inequalities hold:
\beq
\frac{1}{2}\left(1-\sqrt{1-{C^\cap(\rho)}^2}\right)\leq E^\cap_G(\rho)\leq\frac{C^\cap(\rho)}{2}
\eeq
\end{lemma}
\begin{proof}
See Appendix \ref{app:egeoconc}.
\end{proof}
By using the lemma together with the relation \eqref{eq:connection2} for $d=2$, and \eqref{eq:Ccapfidelity}, we immediately obtain the following result:
\begin{theorem} \emph{(Concurrence bounds)}
\label{thm:conc-bnd}
If $J(\Lambda)$ is the Choi-Jamio{\l}kowski state isomorphic to
the qubit channel $\Lambda$, then
\begin{multline}
\frac{1}{2-F(J({\Lambda}),\tilde{J}({\Lambda}))}
\ge
\psucc(\Lambda)
\\\ge\frac{1}{1+\sqrt{1-F(J({\Lambda}),\tilde{J}({\Lambda}))^{2}}}.\label{eq:conc-bnd-p}
\end{multline}
\end{theorem}

\subsubsection{Qudit channels and entanglement of assistance}

In the previous section we focused on the concurrence because its concave-roof version $C^\cap$ can be easily evaluated. Concurrence was generalized to higher-dimensional systems in a number of different ways~\cite{Audenaert:2001p2807,Rungta:2001p2810,Gour:2005p2813}, and even high-dimensional ``assisted'' versions---i.e., concave-roof constructions---were considered~\cite{Gour:2005p2690}. As we mentioned, the most studied example of concave-roof construction is the entanglement of assistance \eqref{eq:Eass}.
For this reason, we will provide bounds for the probability of success in terms of the entanglement of assistance.

We will use the following generalization of the binary entropy that depends only on the number of possible outcomes, $d$, and one probability parameter, $p$:
\beq
h_d(p):=-p \log_2 p - (1-p) \log_2 \left(\frac{1-p}{d-1}\right).
\eeq
That is, $h_d(p)$ is the Shannon entropy of the probability distribution of $d$ symbols $(p,\frac{1-p}{d-1},\dots,\frac{1-p}{d-1})$, with one symbol having probability $p$ and the remaining $d-1$ symbols being equally likely. It coincides with the binary entropy for $d=2$. We remark that $h_d(p)$ is a concave function of $p$, and is monotonically decreasing for $p\geq1/d$. This means that the inverse function $h_d^{-1}:[0,\log_2d]\rightarrow[1/d,1]$ is well defined.

We are now ready to state the theorem that links entanglement of assistance and probability of success.
\begin{theorem}\emph{(Entanglement-of-assistance bounds)}
\label{e-assist-bnd}
For a given qudit channel $\Lambda$, the following inequalities hold:
\beq
\frac{2^{E_{\textup{a}}(J_{\Lambda})}}{d}\geq\psucc(\Lambda)\geq \frac{1}{dh_{d}^{-1}(E_{\textup{a}}(J_{\Lambda}))}
\eeq
where $ E_{\textup{a}}$ is the entanglement of assistance, $J({\Lambda})$ is the Choi-Jamio{\l}kowski state isomorphic to the channel, and $\sigma(\Lambda)$ is the stochasticity of the channel.
\end{theorem}

\begin{proof}
See Appendix \ref{app:eass}.
\end{proof}

\section{Examples}
\label{sec:examples}

In this section, we consider two examples for the qubit case: (i) the amplitude-damping channel
and (ii) the probabilistic constant-output map. For these examples we are able to find analytic
results for the probability of success, and we compare these exact results
with the bounds we obtained in Section \ref{sub:Bounds-1}. The analytic results are obtained by using the triangle-inequality bound of Observation~\ref{obs:convbound} and the already remarked fact that any specific Kraus decomposition provides a upper (lower) bound on the stochasticity (optimal probability of success).

\subsection{Amplitude-damping channel}

The qubit amplitude-damping channel is used to model the decay of an excited
state $\ket{1}$ into the ground state $\ket{0}$. With probability $\epsilon$ the
channel causes the de-excitation of the input state. This de-excitation
process is described by the Kraus operator $A_{1}=\sqrt{\epsilon}|0\rangle\langle1|$.
A second Kraus operator guarantees that the process preserves probability, i.e., that the channel is trace-preserving: $A_{2}=|0\rangle\langle0|+\sqrt{1-\epsilon}|1\rangle\langle1|$.  

Using this specific decomposition and Observation~\ref{obs:convbound} we find:
\[
\|\Lambda(I)\| _{\infty}\le\sigma(\Lambda)\le\|A_{1}\|^{2}_{\infty}+\|A_{2}\|^{2}_{\infty}.
\]
For the lower bound, one finds
\begin{eqnarray*}
\|\Lambda(I)\|_\infty & = & \|A_{1}A_{1}^{\dagger}+A_{2}A_{2}^{\dagger}\|_\infty\\
 & = & \left\|\left(\begin{array}{cc}
1 & 0\\
0 & 1-\epsilon\end{array}\right)+\left(\begin{array}{cc}
\epsilon & 0\\
0 & 0\end{array}\right)\right\|_\infty\\
 & = & 1+\epsilon,
 \end{eqnarray*}
and for the upper bound we get:
\begin{eqnarray*}
\|A_{1}\|^{2}_{\infty}+\|A_{2}\|^{2}_{\infty} & = & 1+\epsilon.\end{eqnarray*}
Because these two bounds coincide, the Kraus decomposition is optimal, with a stochasticity of $\sigma(\Lambda)=1+\epsilon$ and the optimal success probability $\psucc(\Lambda)=(1+\epsilon)^{-1}$.
When $\epsilon=0$ the channel is trivially the identity channel and can be performed deterministically.
However at the other extreme, $\epsilon=1$, the channel becomes the constant
map $\Lambda(\rho)=\Tr(\rho)\proj{0}$, and it can only be realized with  probability $50\%$. As we
found in Section \ref{sub:extreme-pnts}, the constant map is the map
that has the lowest success probability in our scheme. Therefore the
parameter $\epsilon$ that describes the probability of de-excitation
let us move from one extreme to the other of the stochasticity (or probability of success).
For the amplitude-damping channel we find that the concave-roof of concurrence \eqref{eq:concconc} satisfies $C^{\cap}(J(\Lambda))=F(J(\Lambda),\tilde{J}(\Lambda))=\sqrt{1-\epsilon}$. We can then compare the analytic bounds from \eqref{eq:conc-bnd-p} with the exact result we just found (see Figure \ref{fig:AD}):
\beq
\frac{1}{2-\sqrt{1-\epsilon}}\ge \psucc(\Lambda)=\frac{1}{1+\epsilon}\ge\frac{1}{1+\sqrt{\epsilon}}.
\label{eq:bounds}
\eeq
\begin{figure}
\begin{center}
\includegraphics[scale=0.5]{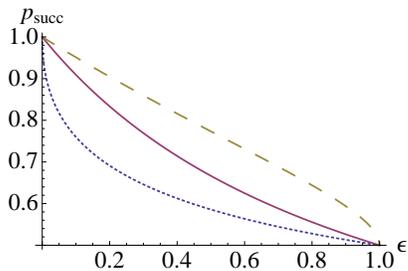}
\end{center}
\caption{Optimal probability of success for the amplitude damping channel (color online): exact (solid line), and upper and lower bounds in Eq. (\ref{eq:bounds}) (dashed and fine-dashed lines, respectively).}
\label{fig:AD}
\end{figure}
We remark that our scheme achieves a probability of success that depends on $\epsilon$ and is close to $100\%$ for $\epsilon$ small. On the contrary, the scheme of~\cite{Qing:2007p1454} has a $50\%$ probability of success, independently of $\epsilon$. In our scheme, such a low success probability is just the worst case ($\epsilon=1$).

\subsection{Probabilistic constant-output channel}

The second channel that we choose to analyze is a convex combination of the constant
output channel and the identity map. Such a channel returns the
input state with a probability of $1-p$ or a fixed output state
$\tau$ with probability $p$. The map is then
\[
\Lambda:\rho\mapsto(1-p)\rho+p\:\Tr(\rho)\tau
\]
and its Choi-Jamio{\l}kowski isomorphic state is simply
\[
J({\Lambda})=(1-p)\proj{\psi_{2}^{+}}+p\,\tau\otimes \frac{I}{2}.
\]
We find the stochasticity of this channel by checking that the upper and
lower bounds for the stochasticity that are generated from Observation 1 and a particular decomposition of the
state match. Without loss of generality we can consider $\tau$ to be diagonal in the computational basis, i.e.,
\beq
\label{eq:tau}
\tau=s|0\rangle\langle0|+\left(1-s\right)|1\rangle\langle1|.
\eeq
In fact, only the degree of mixedness of $\tau$ and not the specific basis influences the probability of success $\psucc$. This can be understood at the formal level by considering that a different choice of basis for $\tau$ can be taken into account via a rotation $U$, which does not influence the entanglement properties of $J(\Lambda)$ (see Appendix \ref{app:invariance}).

Using this state, the bound stated in Observation \ref{obs:convbound} becomes:
\beq
\label{eq:stocfixed}
\sigma(\Lambda)\ge\|\Lambda(I)\|_\infty=\|(1-p)I+2p\tau\|_\infty=1-p+2ps,
\eeq
where we have assumed, without loss of generality, that $s\ge 1/2$. From this and by using Eq.~\eqref{eq:relent}, we find $E_G^\cap(J(\Lambda))\leq 1/2-p(s-1/2)$. One can find an ensemble decomposition of $J(\Lambda)$ that saturates the latter inequality (see Appendix \ref{app:saturation}), therefore $\psucc = (1-p+2ps)^{-1}$.

This means that for this channel we also find that as the probability parameter $p$ varies from $0$ to $1$ we move from the identity map to a constant map. However, we can see from Figure \ref{fig:constout} 
that the success probability of the constant map depends on how mixed the output state is. As expected from the discussion of Section \ref{sub:extreme-pnts}, the lowest value for the success probability, $\psucc(\Lambda)=1/2$, is only attained when the constant output state is pure.

\begin{figure}
\begin{center}
\includegraphics[scale=0.5]{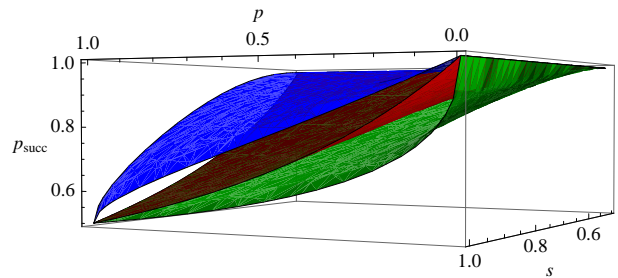}
\end{center}
\caption{ Optimal probability of success for the probabilistic constant-output channel, exact result and bounds (color online). Upper and lower surfaces (blue and green online, respectively) are the upper and lower concurrence bounds from Eq. (\ref{eq:conc-bnd-p}), evaluated for the Choi-Jamio{\l}kowski state isomorphic to the channel; the in-between surface (red online) is the exact success probability $\psucc$ for realizing the channel. Because $\psucc$ is symmetric about the $s=1/2$ plane we only plot this function for $s \in[1/2,1]$. }
\label{fig:constout}
\end{figure}

\section{Discussion}

We have provided a scheme to realize an arbitrary channel on a $d$-rail-encoded optical qudit, taking into account practical restrictions. In particular: (i) we only allow for operations that are realizable with high fidelity using linear optics; (ii) we only allow ancillary modes that are initially in the vacuum state, thus limiting the need for sources of single photons that are, as of now, still difficult to produce on demand; (iii) we do not allow feed-forward (i.e., adaptive schemes), which significantly reduces the cost of the necessary equipment and loss that is inevitably involved in such schemes due to the need of long fibers for optical delays; (iv) and we consider only photon-number measurements (actually, readily available threshold detectors suffice). The conventional linear optics toolbox (phase shifters and beamsplitters), as well as the possibility for randomly switching between different optical networks, are the only elements needed for the realization of our method. These restrictions render our technique of immediate interest to linear-optical implementations that can be realized using state-of-the-art experimental techniques. Within this framework, it turns out that any channel can, in principle, be realized perfectly, albeit only stochastically. The only channels that can be realized deterministically are random-unitary channels. Given that post-selection is a commonly used technique in linear-optics experiments, this restriction effectively only slightly reduces the success probability of an experimental realization, and we are able to provide an expression for the optimal probability of success. This probability turns out to be related to the entanglement properties of the Choi-Jamio{\l}kowski state isomorphic to the channel of interest. More precisely, we were led to evaluate the ``assisted version'' of the geometric measure of entanglement, i.e., the concave-roof extension of the measure to mixed states, for such a state. While we are not aware of a closed formula for it, not even for two-qubit states, we were able to provide upper and lower bounds in terms of the concave-roof of concurrence (for qubits) and of entanglement of assistance (for general qudits).

Besides tackling the problem of evaluating, in general, the concave-roof of the geometric measure of entanglement, i.e., the probability of successful realization of our scheme with the restrictions considered in this paper, future research will focus on the relaxation of said restrictions, that is, on the analysis of more general schemes for the realization of channels.

For example, the use of ancillary states that are not initially in the vacuum certainly improves the realization of certain channels. Indeed, we saw that the worst-case scenario is that of a channel with a fixed---i.e., independent of the input---pure output. We argued that the difficulty---that is, the low probability of success---in the realization of such a channel is essentially due to the necessity of using for the output the same single photon by which the input logical state is encoded in the $d$ modes. This is exactly because no photons are available in the ancillary ports. Obviously, if such a fixed pure output is readily available as ancillary state, the realization of the pure-fixed-output channel becomes trivial: the ancillary input state becomes the output. It is therefore evident that introducing non-vacuum ancillas would strongly affect the performance of our scheme.
Another addition that we plan to consider is feed-forward, that is becoming a powerful and reliable tool in linear-optics quantum information processing~\cite{prevedel07a, Biggerstaff2009a}.

Another possible line of research is that of focusing  on channels that are linked to interesting effects in quantum information processing. Indeed, our results can be thought of as a toolbox to be used in any optical experiment where some specific channel has to be applied, be it for the sake of simulating noise or for implementing a specific protocol. 

We expect our findings to trigger further theoretical studies on
channel realization. In particular, we linked channel realization with more
abstract notions of entanglement theory, and we hope that the study of less explored entanglement properties of states will consequently be stimulated. From a more practical point of view, our results
provide a simple method for realizing arbitrary quantum channels using
linear optics and standard experimental techniques. Our results are ideal for experimental implementation relying on linear optics in combination with post-selection. While quantum channels have been a widely discussed topic in theoretical quantum information, we expect our work to trigger an increased interest in the experimental study of this intriguing topic.

We thank Antonio Ac\'in, Dominic Berry, Robert Prevedel, Kevin Resch and Armin Uhlmann for helpful discussions. We acknowledge support from NSERC, QuantumWorks, Ontario Centres of Excellence, CFI, Ontario MRI and ERA.


\appendix

\section{Proof of Lemma \ref{lem:egeoconc}}
\label{app:egeoconc}

\begin{proof}
By substituting \eqref{eq:relgeoconc} in the definition \eqref{eq:EGcap} of $E^\cap_G(\rho)$ we obtain $E_G^\cap(\rho)=\max_{\{p_i,\psi_i\}} \sum p_{i}\frac{1}{2}\left(1-\sqrt{1-C(\psi_{i})^{2}}\right)$. Then, in order to obtain the lower bound it is sufficient to observe that $\sqrt{1-x^2}$ is a concave  function in $x$ that is monotonically decreasing:
\beq
\begin{split}
E_G^\cap(\rho) & = \max_{\{p_i,\psi_i\}} \sum p_{i}\frac{1}{2}\left(1-\sqrt{1-C(\psi_{i})^{2}}\right)\\
& \geq \max_{\{p_i,\psi_i\}} \frac{1}{2}\left(1-\sqrt{1-\left[\sum p_{i}C(\psi_{i})\right]^{2}}\right)\\
 & = \frac{1}{2}\left(1-\sqrt{1-\left[ \max_{\{p_i,\psi_i\}}\sum p_{i}C(\psi_{i})\right]^{2}}\right)\\
 & = \frac{1}{2}\left(1-\sqrt{1-{C^\cap(\rho)}^2}\right).
 \end{split}
 \eeq
The upper bound can be derived from the relation $\sqrt{1-x^2}\geq 1-x$. 
\end{proof}

\section{Proof of Theorem \ref{e-assist-bnd}}
\label{app:eass}

\begin{proof}
We will use properties of the Shannon entropy $H(\{r_i\})=-\sum_i r_i \log_2 r_i$, defined for a probability distribution $\{r_i\}$. The von Neumann entropy of a quantum state $\rho$ is equal to the Shannon entropy of its eigenvalues $r_i$. In particular, for a pure bipartite state with Schmidt decomposition $\ket{{\psi}}^{AB}=\sum_i\sqrt{\lambda_i}\ket{i}_A\ket{i}_B$, the entropy of the reduced one-party states $\rho^A$ and $\rho^B$ is $H(\{\lambda_i\})$. For any pure ensemble $\{p_a,\psi^{AB}_a\}$ we will denote by $\{\lambda_i^a\}_i$ the set of the squares of the Schmidt coefficients of $\psi^{AB}_a$, and define $\lambda_{a,\textup{max}}=\max_{i}\{\lambda_i^a\}_i$. For the entanglement of assistance it then holds
\begin{multline}
E_{\textup{a}}(J_{\Lambda})\\
\begin{aligned}
& = \max_{\{p_{a},\psi_{a}\}}\sum_{a}p_{a}S(\rho^A_a)\\
& = \max_{\{p_{a},\psi_{a}\}}\sum_{a}p_{a}H(\{\lambda^a_{i}\}_i)\\
 & \le \max_{\{p_{a},\psi_{a}\}}\sum p_{a}H\left(\lambda_{a,\textup{max}},\frac{1-\lambda_{a,\textup{max}}}{d-1},\ldots,\frac{1-\lambda_{a,\textup{max}}}{d-1}\right)\label{eq:entas_up1}\\
 & \le \max_{\{p_{a},\psi_{a}\}} h_{d}\left(\sum_{a}p_{a}\lambda_{a,\textup{max}}\right)\\
 & =  h_{d}\left(\min_{\{p_{a},\psi_{a}\}}\sum_{a}p_{a}\lambda_{a,\textup{max}}\right)\\
 & = h_{d}\left(\frac{\sigma(\Lambda)}{d}\right).
\end{aligned}
 \end{multline}
The first inequality is due to the fact that substituting any subset of probabilities of some distribution with equally weighted probabilities can only increase the total Shannon entropy. This is easily checked by knowing that the flat probability distribution is the one with highest Shannon entropy, and that for any grouping of probabilities $\{r_i\}$ into two subsets $\{r^{(1)}_i\}$ and $\{r^{(2)}_i\}$ of weight $q$ and $1-q$, respectively, we have $H(\{r_i\})=h_2(q)+qH(\{r^{(1)}_i/q\})+(1-q) H(\{r^{(2)}_i/(1-q)\})$. The second inequality is due to the concavity of entropy. The second to last equality is due to the the monotonicity of $h_d$ in the interval $[1/d,1]$. Indeed, $\sum_{a}p_{a}\lambda_{a,\textup{max}}\geq 1/d$ because $\lambda_{a,\textup{max}}\geq 1/d$ for all $a$. Finally, the last equality comes from the relation \eqref{eq:stocschmidt}.
Thus, using the fact that $h_d$ is invertible and monotonically decreasing in the range of interest, we obtain $\sigma(\Lambda)\leq dh_{d}^{-1}(E_{\textup{a}}(J_{\Lambda}))$, i.e., $\psucc\geq1/\left[dh_{d}^{-1}(E_{\textup{a}}(J_{\Lambda}))\right]$.

For the upper bound we have
\begin{eqnarray}
E_{\textup{a}}(\rho_{\Lambda}) & = & \max_{\{p_{a},\psi_{a}\}}\sum p_{a}H(\{\lambda^a_{i}\}_i)\nonumber \\
 & \ge & \max_{\{p_{a},\psi_{a}\}}\sum_{a}\left(-p_{a}\log_2\left(\lambda_{a,\textup{max}}\right)\right)\label{eq:ent-ass-lbnd2}\\
 & \ge & \max_{\{p_{a},\psi_{a}\}}\left(-\log_2\left(\sum p_{a}\lambda_{a,\textup{max}}\right)\right)\label{eq:ent-ass-lbnd1}\\
 & = & -\log_2\left(\min_{\{p_{a},\psi_{a}\}}\sum p_{a}\lambda_{a,\textup{max}}\right)\nonumber \\
 & = & -\log_2\left(\frac{\sigma(\Lambda)}{d}\right).\nonumber \end{eqnarray}
 The first inequality comes from the fact that the min-entropy $H_{\textup{min}}(\{r_i\})=-\log_2 r_{\textup{max}}$ of a probability distribution $\{r_i\}$, with $r_{\textup{max}}=\max\{r_i\}$, satisfies $H_{\textup{min}}(\{r_i\})\leq H(\{r_i\})$. The second inequality is due to the concavity of the logarithm. The second-to-last equality is due to the monotonicity of the logarithm. We finally arrive at the desired relation by exponentiation.
\end{proof}

\section{Basis independence for the probabilistic constant-output channel}
\label{app:invariance}

Suppose $\tau'=U\tau U^\dagger$;
then
\begin{multline*}
J(\Lambda)\\
\begin{aligned}
&=(1-p)\proj{\psi_{d}^{+}}+p\,\tau'\otimes \frac{I}{2}\\
	&=(1-p)\proj{\psi_{d}^{+}}+p\,U\tau U^\dagger\otimes \frac{I}{2}\\
	&=(U\otimes U^*)\left[(1-p)\proj{\psi_{d}^{+}}+p\,\tau \otimes \frac{I}{2}\right](U\otimes U^*)^\dagger,\\	
\end{aligned}
\end{multline*}
where we have used the invariance of the maximally entangled state $\ket{\psi_d^+}=(U\otimes U^*)\ket{\psi_d^+}$, valid for all unitaries $U$.

\section{Decomposition saturating the bound \eqref{eq:stocfixed}}
\label{app:saturation}

One can write
the Choi-Jamio{\l}kowski state as the convex combination
\beq
\begin{split}
J(\Lambda)&=(1-p)\proj{\psi_{2}^{+}}+2p(1-s)\frac{I}{2}\otimes\frac{I}{2}\\
			&+p\big(s-(1-s)\big)|0\rangle\langle0|\otimes \frac{I}{2},
\end{split}
\eeq
such that for the concave-roof of the geometric measure we find
\[
\begin{aligned}
E_G^\cap(J(\Lambda) &\geq (1-p)E^\cap_G(\proj{\psi_{2}^{+}})\\
			&+2p(1-s)E_G^\cap\left(\frac{I}{2}\otimes\frac{I}{2}\right)\\
			&+p\big(s-(1-s)\big)E_G^\cap\left(|0\rangle\langle0|\otimes \frac{I}{2}\right)\\
			&= (1-p) \frac{1}{2}+ 2p(1-s) \frac{1}{2}\\
			&= 1/2-p(s-1/2).
\end{aligned}
\]
Here we used that fact that $E_G^\cap\left(|0\rangle\langle0|\otimes I/2\right)=0$---see the discussion just after Eq. \eqref{eq:extremes}---and that $E^\cap_G\left(I/2\otimes I/2\right)=1/2$, because the maximally mixed state of two qubits can be seen as the convex combination of pure maximally entangled states.

\end{document}